\documentclass[a4paper]{article}
\usepackage[a4paper,includeheadfoot,margin=2.54cm]{geometry}

\usepackage[T1]{fontenc}
\usepackage{amsthm}
\usepackage{amsmath}
\usepackage{amssymb}
\usepackage[affil-it]{authblk}
\usepackage{microtype}
\usepackage{cite}
\usepackage{soul}
\usepackage{xspace}

\usepackage{color}

\usepackage{etoolbox}
\usepackage{lmodern}
\makeatletter
\patchcmd{\@maketitle}{\LARGE \@title}{\fontsize{16}{19.2}\selectfont\@title}{}{}\makeatother

\usepackage{graphicx}
\graphicspath{{./figures/}}

\usepackage[noabbrev,capitalise]{cleveref}
\usepackage{graphbox}

\newcommand{\cR}{\mathcal{R}}
\newcommand{\cG}{\mathcal{G}}
\newcommand{\eps}{\varepsilon}
\DeclareMathOperator{\Lab}{Lab}
\DeclareMathOperator{\lab}{lab}
\DeclareMathOperator{\Tab}{Tab}
\DeclareMathOperator{\nd}{nd}
\DeclareMathOperator{\fd}{fd}
\DeclareMathOperator{\bd}{bd}
\DeclareMathOperator{\td}{td}
\DeclareMathOperator{\id}{id}
\DeclareMathOperator{\br}{br}
\DeclareMathOperator{\cv}{cv}
\DeclareMathOperator{\argmin}{argmin}
\DeclareMathOperator{\argmax}{argmax}

\theoremstyle{plain}
\newtheorem{theorem}{Theorem}[section]
\newtheorem{lemma}[theorem]{Lemma}
\newtheorem{observation}[theorem]{Observation}

\newenvironment{restate}[1][]{%
    \begingroup%
    \def\thetheorem{\ref{#1}}%
}{%
    \addtocounter{theorem}{-1}%
    \endgroup%
}

\title{Routing in Histograms\footnote{
   MC was supported by JST ERATO Grant Number JPMJER1201 (Japan) and ERC STG 757609.
   JC was supported by ERC STG 757609 and DFG grant MU 3501/1-2.
   MK was supported by MEXT KAKENHI No.~17K12635 and the NSF award CCF-1422311.
   WM was supported in part by ERC StG 757609.
   AvR and MR were supported by JST ERATO Grant Number JPMJER1201, Japan.
}}

\author[1]{Man-Kwun Chiu}
\author[1]{Jonas Cleve}
\author[1]{Katharina Klost}
\author[2]{Matias Korman}
\author[1]{Wolfgang Mulzer}
\author[3]{\\Andr\'{e} van Renssen}
\author[4]{Marcel Roeloffzen}
\author[1]{Max Willert}

\affil[1]{Institut f\"ur Informatik, Freie Universit\"at Berlin, Germany\\
\texttt{\{chiumk,jonascleve,kathklost,mulzer,willerma\}@inf.fu-berlin.de}}
\affil[2]{Department of Computer Science, Tufts University, Medford, MA, USA\\
  \texttt{matias.korman@tufts.edu}}
\affil[3]{School of Computer Science, University of Sydney, Sydney, Australia\\
  \texttt{andre.vanrenssen@sydney.edu.au}}
\affil[4]{Department of Mathematics and Computer Science, TU Eindhoven, 
Eindhoven, the Netherlands\\
\texttt{m.j.m.roeloffzen@tue.nl}}

\date{}
\begin{document}
\maketitle

\begin{abstract}
Let $P$ be an $x$-monotone orthogonal polygon with $n$ vertices.
We call $P$ a \emph{simple histogram} if its upper boundary
is a single edge; and a \emph{double histogram} if it has
a horizontal chord from the left boundary to the right 
boundary.
Two points $p$ and $q$ in $P$ are \emph{co-visible} if and only
if the (axis-parallel) rectangle spanned by $p$ and $q$ 
completely lies in $P$.
In the $r$-visibility graph $G(P)$ of $P$, we connect two vertices
of $P$ with an edge if and only if they are co-visible.

We consider \emph{routing with preprocessing} in $G(P)$.
We may preprocess $P$ to obtain 
a \emph{label} and a \emph{routing table} for each vertex of $P$. 
Then, we must be able to route a packet between
any two vertices $s$ and $t$ of $P$, where each step 
may use only the label of the target node $t$, the routing table 
and neighborhood of the current node, and 
the packet header.

We present a routing scheme for double histograms that sends 
any data packet along a path whose length is at most twice
the (unweighted) shortest path distance between the 
endpoints.
In our scheme, the labels, routing tables, and headers
need $O(\log n)$ bits. For the case of simple histograms,
we obtain a routing scheme with optimal routing paths,
$O(\log n)$-bit labels, one-bit routing tables, 
and no headers.
\end{abstract}

\section{Introduction}

The \emph{routing} problem is a classic question in 
distributed graph algorithms~\cite{GiordanoSt04,PelegUp89}. 
We have a graph $G$ and would like to preprocess it
for the following task: given a data
packet located at some \emph{source} vertex $s$ of $G$,
route the packet to a \emph{target} vertex $t$ of $G$,
identified by its \emph{label}.
The routing should have the following properties:
(A) \emph{locality}: to determine
the next step of the packet, it should use only 
information at the current vertex or in the 
packet header; (B) \emph{efficiency}:
the packet should travel along a path whose
length is not much larger than the length of a shortest path
between $s$ and $t$. The ratio between the length of this
routing path and a shortest 
path is called the \emph{stretch factor}; and (C)
\emph{compactness}: the space
requirements for labels, routing tables, and packet headers
should be small.

Obviously, we could store at each vertex $v$ of $G$
the complete shortest path tree of $v$. Then, the routing scheme
is perfectly efficient: we can send the packet
along a shortest path. However, the scheme lacks compactness.
Thus, the challenge is to balance the
(seemingly) conflicting goals of compactness and efficiency.

There are many compact routing schemes for general graphs
\cite{AbrahamGa11,AwerbuchBNLiPe90,Chechik13,Cowen01,
EilamGaPe03,RodittyTo15,RodittyTo16}.
For example, the scheme by  Roditty and Tov~\cite{RodittyTo16} 
needs to store a 
poly-logarithmic number of bits in the packet header and
it routes a packet from $s$ to $t$ on a path of length 
$O\big(k\Delta+m^{1/k}\big)$, where $\Delta$ is the 
shortest path distance between $s$ and $t$, $k > 2$ is any 
fixed integer, $n$ is the number of nodes, and $m$ is the number of
edges. 
The routing tables use $mn^{O(1/\sqrt{\log n})}$ space.
In the late 1980's, Peleg and Upfal~\cite{PelegUp89} proved
that in general graphs, any routing scheme with constant stretch factor
must store $\Omega(n^c)$ bits per vertex, for some constant 
$c > 0$.
Thus, it is natural to focus on special graph classes to obtain
better routing schemes. For 
instance, trees admit routing schemes that always 
follow the shortest path and that store $O(\log n)$ bits at 
each node~\cite{FraigniaudGa01,SantoroKh85,ThorupZw01}. 
Moreover, in planar graphs,
for any fixed $\eps > 0$, there is a routing scheme with a 
poly-logarithmic number of bits in each routing table that 
always finds a path that is within a factor of 
$1 + \eps$ from optimal~\cite{Thorup04}.
Similar results are also available for
unit disk graphs~\cite{KaplanMuRoSe18,yan2012compact}
and for metric spaces with 
bounded doubling dimension~\cite{KonjevodRiXi16}.

Another approach is \emph{geometric routing}:
the graph resides in a geometric space, and 
the routing algorithm has to determine the next vertex for 
the packet based on the coordinates of the source and the target 
vertex, the current vertex, and its neighborhood, see 
for instance~\cite{BoseFavReVe17,BoseFavReVe15} and the references 
therein. In contrast to compact routing schemes, there are no 
routing tables, and the routing happens purely based on the 
local geometric information (and possibly the packet header).
For example, the routing 
algorithm for triangulations by Bose and Morin~\cite{BoseMo04} uses 
the line segment between the source and the target for its 
routing decisions. In a recent result, 
Bose~{\em et al.}~\cite{BoseFavReVe17} 
show that when vertices do not store any routing tables,
no geometric routing scheme can achieve the stretch 
factor $o(\sqrt{n})$. This lower bound applies irrespective of the
header size.

We consider routing in a particularly interesting class
of geometric graphs, namely visibility graphs of polygons.
Banyassady~{\em et al.}~\cite{banyassady2017routing} presented a
routing scheme for
polygonal domains with $n$ vertices and $h$ holes that uses $O(\log n)$
bits for the label, $O((\eps^{-1}+h) \log n)$ bits for the routing tables,
and achieves a stretch of $1 + \eps$, for any fixed $\eps>0$.
However, their approach is efficient only  if 
the edges of the visibility graph are
weighted with their Euclidean lengths.
Banyassady~{\em et al.}~ask whether  there is 
an efficient routing scheme for visibility graphs with unit weights
(also called the \emph{hop-distance}), arguably a more applied setting.

We address this open problem by combining
the two approaches of geometric and compact routing: we use
routing tables at the vertices to represent information about the 
structure of the graph, 
but we also assume that the labels of all adjacent vertices are directly
visible at each node.
This is reasonable from a practical point of view, because a 
node in a network must be aware of all its neighbors and their 
labels.
The size of this list is not relevant for the compactness,
since it depends purely on the graph and cannot be influenced 
during preprocessing.
We focus our attention on $r$-visibility graphs of orthogonal 
simple and double histograms. Even this seemingly simple case turns 
out to be quite challenging and reveals the whole richness of 
the compact routing problem in unweighted, geometrically defined 
graphs. Furthermore, histograms 
constitute a natural starting point, 
since they are crucial building blocks in many visibility 
problems; see, for instance,~\cite{bartschi2011coloring, bartschi2014improved,bartschi2014conflict,bhattacharya2017approximability,bhattacharya2017constant,hoffmann2018tight}.
In addition, $r$-visibility is a popular concept in orthogonal 
polygons that enjoys many useful structural properties,
see, e.g.,~\cite{o1987art,hoffmann2018tight,motwani1990covering,hoffmann1990rectilinear,worman2007polygon}.

A simple histogram is a monotone orthogonal
polygon whose upper boundary consists of a single edge;
a double histogram is a monotone orthogonal polygon
that has a horizontal chord that touches the boundary of $P$ only
at the left and the right boundary.
Let $P$ be a (simple or double) histogram with $n$ vertices.
Two vertices $v$ and $w$ in $P$ are connected in the visibility 
graph $G(P)$ by an unweighted edge if
and only if 
the axis-parallel rectangle spanned by $v$ and $w$ is contained in the 
(closed) region $P$ (we say that $v$ and $w$ are \emph{co-visible}). 
We present the first efficient and compact routing schemes for 
polygonal domains under the hop-distance. The following two theorems 
give the precise statements.
\begingroup
\def\thetheorem{1}
\begin{theorem}\label{thm:simple-routing}
Let $P$ be a simple histogram with $n$ vertices. There is a routing scheme
for $G(P)$ with a routing table with 1 bit, without headers, having label
size $2\cdot\lceil\log n\rceil$, such that we can route between any two
vertices on a shortest path.
\end{theorem}

\def\thetheorem{2}
\begin{theorem}\label{thm:double-routing}
Let $P$ be a double histogram with $n$ vertices. There is a routing
scheme for $G(P)$ with routing table, label and header size
$O(\log n)$, such that we can route between any two vertices
with stretch at most $2$.
\end{theorem}
\endgroup

\section{Preliminaries}

\subparagraph*{Routing schemes.}
Let $G = (V, E)$ be an \emph{undirected, unweighted, 
simple, connected graph}. The (closed) \emph{neighborhood} 
of a vertex $v \in V$, $N(v)$, is the set containing $v$ 
and its adjacent nodes.
Let $v, w\in V$. A sequence 
$\pi: \langle v = p_0, p_1, \dots, p_k = w\rangle$ 
of vertices with $p_{i-1}p_i \in E$, for
$i = 1, \dots, k$, is called a \emph{path} of 
length $k$ between $v$ and $w$. The length of 
$\pi$ is denoted by $|\pi|$. 
We define $d(v,w) = \min_{\pi}|\pi|$ as 
the length of a shortest path between $v$ and $w$, 
where $\pi$ goes over all paths with
endpoints $v$ and $w$.

Next, we define a \emph{routing scheme}.
The algorithm that decides the next 
step of the packet is modeled by a 
\emph{routing function}. Every node is assigned a (binary) 
\emph{label} that identifies it in the network.
The routing function  
uses local information at the current node, the label of the 
target node, and the \emph{header} stored in the packet.
The local information of a
node $v$ has two parts: (i) the \emph{link table}, 
a list of the labels of $N(v)$ and (ii)
the \emph{routing table}, a bitstring chosen
during preprocessing to represent relevant 
topological properties of $G$.
Formally, a routing scheme of a graph $G$ 
consists of:
\begin{itemize}
\item a label $\lab(v) \in\{0,1\}^+$ for each 
node $v\in V$;
\item a routing table $\rho(v) \in \{0,1\}^*$ for 
each node $v\in V$; and
\item a routing function 
$f \colon \big(\{0,1\}^*\big)^4
\rightarrow V \times \{0,1\}^*$. 
\end{itemize}
The routing function takes the link table 
and routing table of a current node $s \in V$, 
the label $\lab(t)$ of a target node $t$, and a 
header $h \in \{0,1\}^*$. 
Using these four inputs, it
provides a next node $v$ adjacent to $s$ 
and a new header $h'$. The local information in 
the packet is updated to $h'$, and it 
is forwarded to $v$.  The routing scheme 
is \emph{correct} if the following holds: for 
any two sites $s, t\in V$, consider the
sequence $(p_0, h_0)=(s, \eps)$ 
and $(p_{i+1}, h_{i+1}) =
f\Big(\lab\big(N(p_{i})\big), 
\rho(p_{i}), \lab(t), h_{i}\Big)$, for $i \geq 0$. 
Then, there is
a $k = k(s, t) \geq 0$ with $p_k = t$ and 
$p_{i} \neq t$, for $i = 0, \dots, k - 1$. 
The routing scheme \emph{reaches} $t$ 
in $k$ steps. Furthermore, 
$\pi: \langle p_0,\dots, p_k\rangle$ is the 
\emph{routing path} from $s$ to 
$t$. The \emph{routing distance} is
$d_{\rho}(s,t) = |\pi|$.

Now, let $\cR$ be a \emph{family} of correct 
routing schemes for a given graph class 
$\cG$, i.e., $\cR$ contains a correct routing 
scheme for every graph in $\cG$ such that all 
routing functions use the same algorithm.
There are several measures for
the quality of $\cR$.  
First, the various pieces of information used 
for the routing should be small. 
This is measured by the \emph{maximum label size} 
$\Lab(n)$, the \emph{maximum routing table size} 
$\Tab(n)$, and the \emph{maximum header size} $H(n)$,
over all graphs in $\cG$ of a certain size.  
They are defined as

\begin{align*}
  \text{Lab}(n) &= \max_{\substack{(V, E) \in \cG \\
      |V| = n}} \max_{v \in V} |\lab(v)|,\\
  \text{Tab}(n) &=\max_{\substack{(V, E) \in \cG \\
      |V| = n}} \max_{v \in V} |\rho(v)|, \text{ and }\\
  H(n) &= \max_{\substack{(V, E) \in \cG \\
      |V| = n}} \max_{s \neq t\in V}\max_{i = 0, \dots, k(s, t)} |h_i|.
  \intertext{Finally, the \emph{stretch}
     $\zeta(n)$ relates the length of the routing path to the shortest
     path distance:}
  \zeta(n) &= \max_{\substack{(V, E) \in \cG \\
      |V| = n}} \max_{s \neq t\in V}
\frac{d_{\rho}(s, t)}{d(s, t)}.
\end{align*}

\subparagraph*{Polygons.}
Let $P$ be a \emph{simple orthogonal (axis-aligned)
polygon} in 
general position with vertex set $V(P)$, 
$|V(P)| = n$. No three vertices in $V(P)$ are 
on the same vertical or horizontal line.
The vertices are indexed counterclockwise 
from $0$ to $n - 1$; the 
lexicographically largest vertex has index $n - 1$. 
For $v \in V(P)$, we write $v_x$ for the $x$-coordinate,
$v_y$ for the $y$-coordinate, and $v_{\id}$ for the 
index.

We consider \emph{$r$-visibility}: 
two points $p,q\in P$ \emph{see each other} 
(are \emph{co-visible}) if and only if 
the axis-aligned rectangle 
spanned by $p$ and $q$ lies inside $P$ 
(we treat $P$ as a closed set).
The \emph{visibility graph} 
$G(P) = \big(V(P), E(P)\big)$ of $P$ has
an edge between two vertices $v,w\in V(P)$ 
if and only if $v$ and $w$ are co-visible. 
The distance $d(v, w)$ between two vertices 
$v, w \in V(P)$
is called the \emph{hop distance} of $v$ and $w$ in $P$.

A \emph{histogram} is an $x$-monotone orthogonal 
polygon where the upper boundary
consists of exactly one horizontal edge, the 
\emph{base edge}.
Due to our numbering convention, the endpoints of 
the base edge are indexed $0$ (left) and $n-1$ (right).
They are called the \emph{base vertices}.
A \emph{double histogram} is an $x$-monotone orthogonal
polygon $P$ that has a \emph{base line}, 
a horizontal line segment whose relative interior lies in the
interior of $P$ and whose left and right endpoint are on the 
left and right boundary edge of $P$, respectively.
We assume that the base line lies on the $x$-axis. 
Two vertices $v$, $w$ in $P$ lie \emph{on the same side} 
if both are below or above the base line, i.e., 
if $v_y w_y > 0$.
Every histogram is also a double histogram.
From now on, we let $P$ denote a  
(double) histogram.

Next, we classify the vertices 
of $P$.
A vertex $v$ in $P$ is incident to exactly 
one horizontal edge $h$. We call $v$ a 
\emph{left} vertex if it is the left endpoint of $h$; 
otherwise, $v$ is a \emph{right} vertex. 
Furthermore, $v$ is \emph{convex}
if the interior angle at $v$ is $\pi/2$; otherwise,
$v$ is \emph{reflex}. Accordingly, 
every vertex of $P$ is either
\emph{$\ell$-convex}, \emph{r-convex}, 
\emph{$\ell$-reflex}, or
\emph{r-reflex}.

\subparagraph*{Visibility Landmarks.}
To understand the structure of shortest
paths in $P$, we associate with 
each $v \in V(P)$ three landmark points
in $P$ (not necessarily vertices); 
\cref{fig:l-r-cv-I} gives an
illustration.
\begin{figure}
\begin{center}
\includegraphics[scale=0.8]{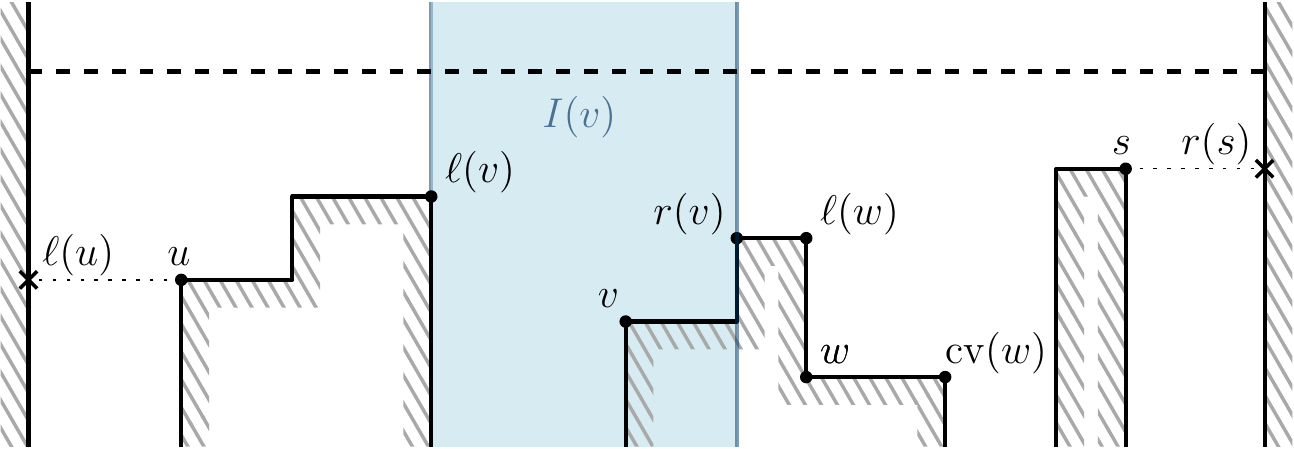}
\end{center}
\caption{The left and right points and the corresponding vertex.
The interval $I(v)$ of $v$ is the set of vertices between $\ell(v)$
and $r(v)$. The dashed line represents the base line.}
\label{fig:l-r-cv-I}
\end{figure}
The \emph{corresponding vertex} of $v$, 
$\cv(v)$, is the unique vertex that shares
a
horizontal edge with $v$.
To obtain the \emph{left point} $\ell(v)$ of $v$, 
we shoot a leftward 
horizontal ray $r$ from $v$. Let $e$ be the vertical 
edge where $r$ first hits the boundary of $P$.
If $e$ is the left boundary of $P$; then if $P$ is a simple
histogram, we let $\ell(v)$ be the left base vertex; and otherwise
$\ell(v)$ is the point where $r$ hits $e$.
If $e$ is not the left boundary of $P$, we let
$\ell(v)$ be the endpoint of $e$ closer to the base line.
The \emph{right point} $r(v)$ of $v$ is defined analogously,
by shooting the horizontal ray to the right.

Let $p$ and $q$ be two points in $P$. We say that $p$
is \emph{to the left of} $q$, if $p_x\leq q_x$. The point $p$ is \emph{strictly
to the left} of $q$, if $p_x<q_x$. The terms \emph{to the right of}
as well as \emph{strictly to the right of} are defined analogously. The 
\emph{interval} $[p, q]$ of $p$ and $q$ 
is the set of vertices in $P$ between $p$ 
and $q$, 
$[p, q] = \big\{v \in V(P) \mid p_x \leq v_x \leq q_x \big\}$.
By general position, this 
corresponds to index intervals
in simple histograms. More precisely, if $P$ is 
a simple histogram and $p$ is either an $r$-reflex vertex or 
the left base vertex and $q$ is either 
$\ell$-reflex or the right base vertex, then
$[p, q] = \big\{v \in V(P) \mid p_{\id} \leq
v_{\id} \leq q_{\id} \big\}$.
The \emph{interval of a vertex $v$}, 
$I(v)$, is the interval of the left and right point
of $v$, $I(v)=[\ell(v), r(v)]$. 
Every vertex visible from $v$ is in $I(v)$, i.e.,
$N(v) \subseteq I(v)$. This interval plays a crucial
role in our routing scheme and gives a very powerful
characterization of visibility in double histograms.

Let $s$ and $t$ be two vertices with
$t \in I(s) \setminus N(s)$.
We define two more landmarks for 
$s$ and $t$. Assume that $t$ lies 
strictly to the right of $s$, the other case is symmetric.
The \emph{near dominator} $\nd(s,t)$
of $t$ with respect to $s$ is the rightmost 
vertex in $N(s)$ to the left of $t$. 
If there is more than one such vertex, 
$\nd(s, t)$
is the vertex closest to the base line. 
Since $t$ is not visible from $s$, the near dominator 
always exists. The \emph{far dominator} $\fd(s,t)$  
of $t$ with respect to $s$ is the leftmost vertex in $N(s)$
to the right of $t$.
If there is more than one such vertex, 
$\fd(s, t)$
is the vertex closest to the base line. If there is no 
such vertex, we set $\fd(s, t) = r(s)$, the projection of 
$s$ on the right boundary.
The interval $I(s,t) = \big[\!\nd(s,t),\fd(s,t)\big]$ has
all vertices between the near and far dominator; see
 \cref{fig:nd-fd}.
\begin{figure}
\begin{center}
\includegraphics[scale=0.8]{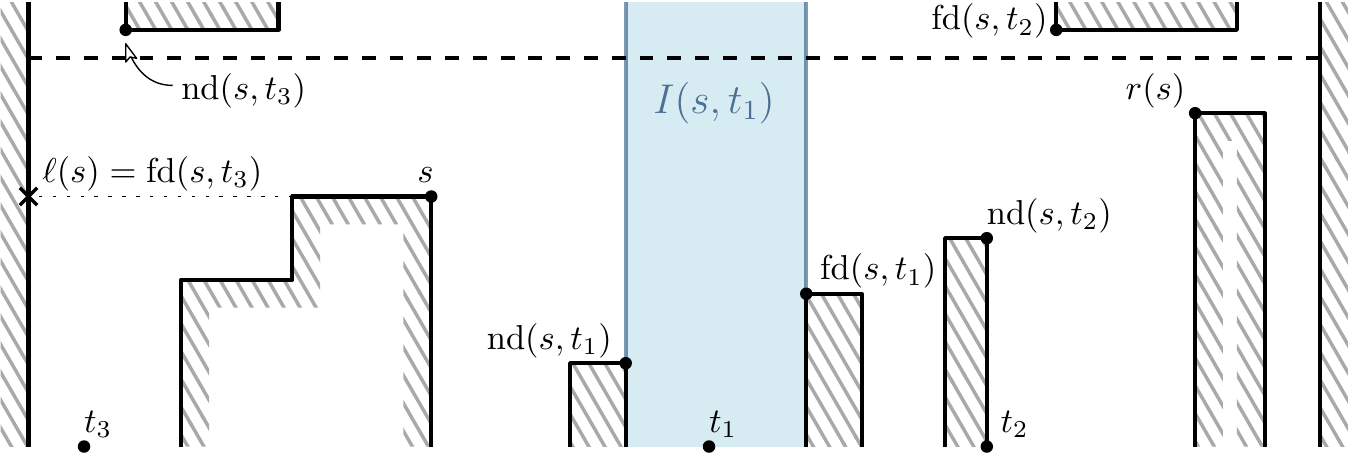}
\end{center}
\caption{The near and the far dominators. Observe that 
$\fd(s,t_3)$ is not a vertex.}
\label{fig:nd-fd}
\end{figure}
\section{Simple Histograms}
Let $P$ be a simple histogram with $n$ vertices. First, we give several
intuitive characterizations of the visibility in $P$.
Then, we analyze how the shortest paths between vertices behave.
The idea for our routing scheme is as follows: as long as a target
vertex $t$ is not contained in the interval $I(s)$ of a current vertex
$s$, i.e., as long as there is a higher vertex that blocks visibility 
between $s$
and $t$, we have to leave the current pocket as far as possible. 
Once we have reached a high enough spike,
we have to find the pocket containing $t$.
Finding the right pocket is possible, but much harder than just going up.
Details follow.

\subsection{Visibility in Simple Histograms}

We begin with some observations on the visibility in $P$.
As $P$ is a simple histogram, 
we have that for all vertices $v \in V(P)$, 
the points $\ell(v)$ and $r(v)$ are vertices of $P$.
Therefore, the far dominators also have
to be vertices. The following observations are now immediate.

\begin{observation}
\label{obs:interval-inclusion}
Let $v \in V(P)$ be $r$-reflex or the left base vertex,
and let $u \in [v, r(v)]$ be a vertex 
distinct from $v$ and $r(v)$. 
Then, $I(u) \subseteq [v, r(v)]$.
\end{observation}
\begin{proof}
Assume $\ell(u)$ or $r(u)$ is outside of $[v, r(v)]$. Then, 
$u$ has a larger $y$-coordinate than $v$. Thus, $v$ cannot 
see $r(v)$, a contradiction to the definition of $r(v)$.
\end{proof}

\begin{observation}
\label{obs:only-2-visible-opposite-site}
Let $v \in V(P)$ be a left \textup(right\textup) vertex 
distinct from the base vertex.
Then, $v$ can see exactly two vertices to its right 
\textup(left\textup): $\cv(v)$ and $r(v)$ \textup($\ell(v)$\textup).
\end{observation}
\begin{proof}
Suppose that $v$ is a left vertex; the other 
case is symmetric.
Any vertex visible from $v$ to the right of $v$ 
lies in in $[\cv(v), r(v)]$. If 
$\cv(v)$ is convex, the observation 
is immediate, since then 
$[\cv(v), r(v)] = \{\cv(v), r(v)\}$. 
Otherwise, $\cv(v)$ is $r$-reflex and 
$r(v) = r(\cv(v))$.
By \cref{obs:interval-inclusion}, we get that 
for all $u \in [\cv(v), r(v)] \setminus 
\{\cv(v), r(v)\}$,  we have
$I(u) \subseteq [\cv(v), r(v)]$.
Thus, $v \notin I(u)$ for any such $u$, and 
since $N(u) \subseteq I(u)$, $v$ cannot see $u$.
\end{proof}

\subsection{Paths in a Simple Histogram}

We now analyze the structure of (shortest) paths 
in a simple histogram. The following lemma 
identifies certain ``bottleneck'' vertices that 
must appear on any path; 
see \cref{fig:any-path-includes-interval-border}.

\begin{lemma}
\label[lemma]{lem:any-path-includes-interval-border}
Let $v, w \in V(P)$ be co-visible vertices such that
$v$ is either $r$-reflex or the left base vertex 
and $w$ is either $\ell$-reflex or the
right base vertex.
Let $s$ and $t$ be two vertices with $s \in [v, w]$ and
$t \notin[v, w]$. Then, any path between $s$ and $t$ 
includes $v$ or $w$.
\end{lemma}

\begin{proof}
Let $I = [v,w]$. Since $t \notin I$, not both
$v, w$ can be base vertices. Thus, suppose
without loss of generality that $v_y < w_y$. 
Then, $\cv(v)$
is a left vertex and can see $w$. Hence, 
\cref{obs:only-2-visible-opposite-site} implies that
$r(v) = r(\cv(v))=w$. By
\cref{obs:interval-inclusion}, we get $N(u) \subseteq I(u) 
\subseteq I$, for every $u \in I \setminus \{v, w\}$. 
Thus, any path between $s$ and $t$ must include $v$ or $w$.
\end{proof}
\begin{figure}[htbp]
\begin{center}
\includegraphics[align=t,scale=0.8]{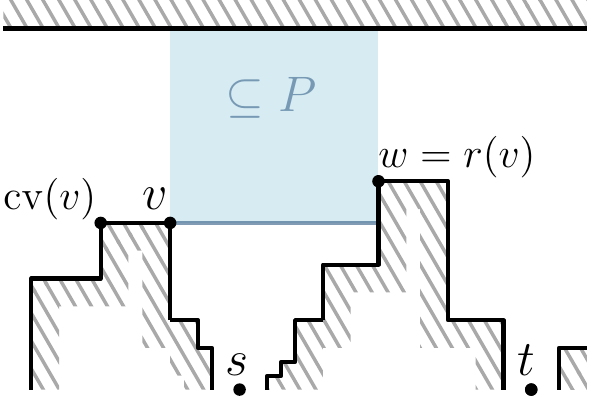}
\hspace*{1cm}
\includegraphics[align=t,scale=0.8]{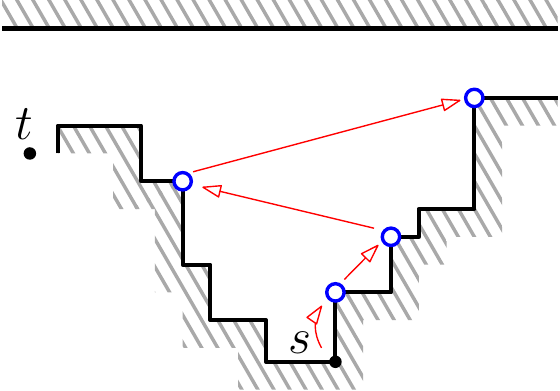}
\end{center}
\caption{Left: Any path from $s$ to $t$ has to include $v$ or $w$, since
the blue rectangle does not contain a vertex except for $v$ and $w$.
Right: A shortest path from $s$ to $t$ using the highest vertex.}
\label{fig:any-path-includes-interval-border}
\end{figure}
An immediate consequence of 
\cref{lem:any-path-includes-interval-border} is
that if $t \notin I(s)$, then any path 
from $s$ to $t$ uses $\ell(s)$ or $r(s)$. The next lemma 
shows that if $t \notin I(s)$, there is a shortest 
path from $s$ to $t$ that uses the higher 
vertex of $\ell(s)$ and $r(s)$, 
see \cref{fig:any-path-includes-interval-border}.

\begin{lemma}
\label[lemma]{lem:taking-higher-is-better}
Let $s$ and $t$ be two vertices with $t \notin I(s)$. 
If $\ell(s)_y > r(s)_y$ \textup($\ell(s)_y<r(s)_y$\textup), 
then there is a shortest path from $s$ to $t$ using 
$\ell(s)$ \textup($r(s)$\textup).
\end{lemma}

\begin{proof}
Assume $\ell(s)_y > r(s)_y$, the other case is symmetric.
Let $\pi: \langle s = p_0, \dots, p_k = t \rangle$ be a shortest 
path from $s$ to $t$. If $\pi$ contains $\ell(s)$, we 
are done. Otherwise, by 
\cref{lem:any-path-includes-interval-border}, there
is a $0 < j< k$ with $p_j = r(s)$ and $p_i \neq r(s)$, 
for $i > j$. Thus, $p_{j + 1} \notin I(s)$. Since 
we assumed $\ell(s)_y > r(s)_y$, it follows that 
$\ell(p_j) = \ell(r(s)) = \ell(s)$, so
$p_{j + 1}$ must be to the right of $p_j$. 
Therefore, by \cref{obs:only-2-visible-opposite-site}, 
we can conclude that
$p_{j + 1} \in \{\cv(p_j),r(p_j)\}$. 
Now, since $\ell(s)$ is
higher than $r(s)$, it can also see 
$\cv(p_j)$ and $r(p_j)$, in particular, 
it can see $p_{j+1}$.
Hence, $\langle s, \ell(s),p_{j+1}, \dots, p_k \rangle$ is 
a valid path of length at most
$|\pi|$, so there exists a shortest path from $s$ to
$t$ through $\ell(s)$.
\end{proof}

The next lemma considers the case where $t$ is in $I(s)$.
Then, the near and far dominator 
are the potential vertices that lie on a shortest
path from $s$ to $t$.

\begin{lemma}
\label[lemma]{lem:nd-sees-fd-simple}
Let $s$ and $t$ be two vertices with
$t \in I(s) \setminus N(s)$. Then,
$\nd(s,t)$ is reflex and either
$\fd(s,t)=\ell(\nd(s,t))$ or
$\fd(s,t)=r(\nd(s,t))$.
\end{lemma}

\begin{proof}
Without loss of generality, $t$ lies strictly to the right 
of $s$. First, assume that $\nd(s,t)$ 
is $\ell$-convex. Since $s$ can
see $\nd(s, t)$ and since $\nd(s,t)$ is to the right of $s$,
it follows that $s$ and 
$\nd(s, t)$ share the same
vertical edge. Then, 
$\cv(\nd(s,t))$ is
also visible from $s$ and its horizontal distance 
to $t$ is smaller. This contradicts the definition 
of $\nd(s,t)$.

Next, assume that $\nd(s,t)$ is $r$-convex. 
Let $v$ be the reflex vertex sharing a 
vertical edge with  $\nd(s,t)$. 
Then, $N(\nd(s, t)) \subseteq N(v)$ and
$v \in N(s)$. Furthermore, since $t$ is strictly to the right
of $v$ but still inside $I(s)$, the vertices $v$ and $r(s)$ must be 
distinct.
Thus, $v_y < s_y$, so that $\cv(v)$ is also visible
from $s$. Moreover, the horizontal distance of 
$\cv(v)$ and $t$ is
smaller than the horizontal distance of $\nd(s,t)$ and 
$t$. This again contradicts the definition of 
$\nd(s,t)$. The first part of the lemma follows.

It remains to show that $\fd(s,t) = r(\nd(s,t))$.
First of all, $\fd(s,t)$ is higher
than $\nd(s,t)$, since otherwise 
$\fd(s,t)$ would not be visible from $s$. Moreover, if 
$\nd(s,t)$ and $\fd(s,t)$ are not 
co-visible, there must be a vertex $v$ strictly between 
$\nd(s,t)$ and $\fd(s,t)$
that is visible from $s$ and higher than $\nd(s,t)$.
Now, either
$t\in[\nd(s,t),v]$ or $t\in[v,\fd(s,t)]$. 
In the first case, the horizontal distance between $v$ and $t$ 
is smaller than between $t$ and $\fd(s, t)$, and in the
second case, the horizontal distance between $v$ and $t$ 
is smaller than between $t$ and $\nd(s, t)$. Either
case leads to a contradiction.
Therefore, $\fd(s,t)$ is higher than $\nd(s,t)$, strictly to
the right of $\nd(s,t)$ and visible from $\nd(s,t)$. Thus,
\cref{obs:only-2-visible-opposite-site} gives $\fd(s,t) = r(\nd(s,t))$.
\end{proof}

\subsection{The Routing Scheme}
\begin{figure}
\begin{center}
\includegraphics[scale=0.9]{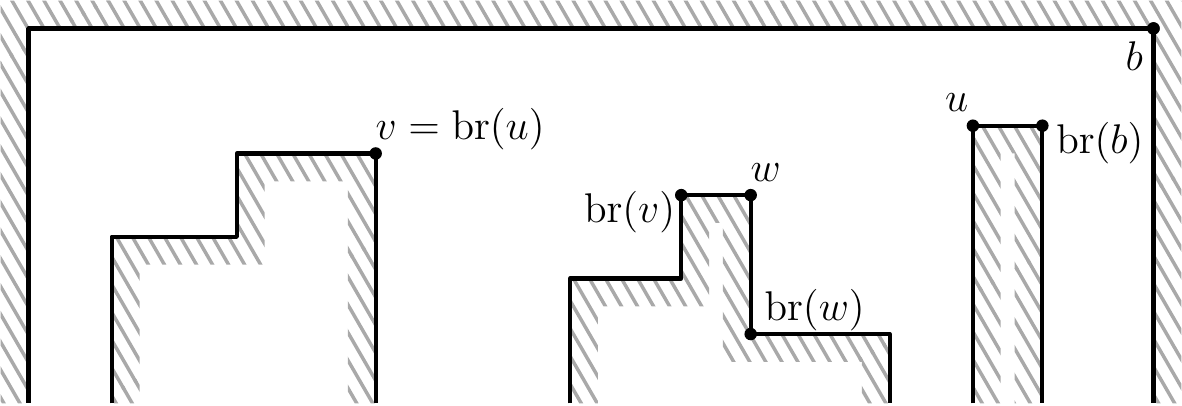}
\end{center}
\caption{The breakpoints of some vertices.}
\label{fig:breakpoint}
\end{figure}
We now describe our routing scheme and prove that it 
gives a shortest path.
\subparagraph*{Labels and routing tables.}
Let $v$ be a vertex. If $v$ is convex and not a
base vertex, it is labeled with its $\id$, i.e.,
$\lab(v)=v_{\id}$. Otherwise, 
suppose that $v$ is an $r$-reflex vertex or the
left base vertex.
The \emph{breakpoint of $v$}, $\br(v)$, is 
defined as the left endpoint of the horizontal 
edge with the highest $y$-coordinate to the right 
of and below $v$ that is visible from $v$; analogous
definitions apply to $\ell$-reflex vertices and the right 
base vertex; see \cref{fig:breakpoint}. 
Then, the label of $v$ consists
of the $\id$s of $v$ and its breakpoint, i.e.,
$\lab(v)=(v_{\id},\br(v)_{\id})$. Therefore,
$\Lab(n) = 2\cdot\lceil\log n\rceil$.
The routing table stores one bit,
indicating whether $\ell(v)_y>r(v)_y$, or not.
Hence, $\Tab(n)=1$.

\subparagraph*{The routing function.}
We are given the current vertex $s$  and the label 
$\lab(t)$ of the target vertex $t$. The routing function
does not use any information from the header, i.e., $H(n) = 0$.
If $t$ is visible from $s$, i.e.,
if $\lab(t) \in \lab(N(s))$, we directly go from 
$s$ to $t$ on a shortest path.
Thus, assume 
that $t$ is not visible from $s$.
First, we check whether $t\in I(s)$.
This is done as follows: we determine the smallest 
and largest $\id$ in the link table $\lab(N(s))$ of $s$.
The corresponding vertices are $\ell(s)$ and $r(s)$. 
Then, we can check whether $t_{\id} \in[\ell(s)_{\id},r(s)_{\id}]$,
which is the case if and only if $t\in I(s)$. Now,
there are two cases, 
illustrated in \cref{fig:routing-cases-simple}.
\begin{figure}
\begin{center}
\includegraphics{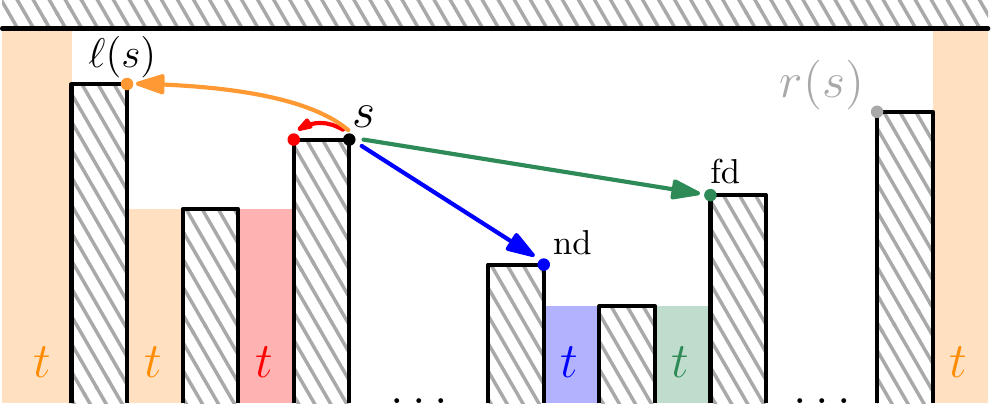}
\end{center}
\caption{The cases where the vertex $t$ lies and the corresponding
vertices where the data packet is sent to. If $t\in [\ell(s),s]$ we
have $\nd(s,t)=\cv(s)$ and
$\fd(s,t)=\ell(s)$.}
\label{fig:routing-cases-simple}
\end{figure}
First, suppose $t \notin I(s)$.
If the bit in the routing table of $s$ indicates that $\ell(s)$ is 
higher than $r(s)$, we take the hop to $\ell(s)$; otherwise,
we take the hop to $r(s)$.
By \cref{lem:taking-higher-is-better}, this hop
lies on a shortest path from $s$ to $t$.

Second, suppose that 
$t \in I(s)\setminus N(s)$.
This case is a bit more involved. We use the link table
$\lab(N(s))$ of $s$ and the label $\lab(t)$ of $t$ to 
determine $\fd(s,t)$ and $\nd(s,t)$. 
Again, we can do this by comparing the $\id$s.
\cref{lem:nd-sees-fd-simple} states that either
$\fd(s,t)=\ell(\nd(s,t))$ or
$\fd(s,t)=r(\nd(s,t))$.
We discuss the case that $\fd(s,t)=
r(\nd(s,t))$, the other case is symmetric. By
\cref{lem:any-path-includes-interval-border},
any shortest path from $s$ to $t$ includes
$\fd(s,t)$ or $\nd(s,t)$. Moreover, due to
\cref{lem:nd-sees-fd-simple}, $\nd(s,t)$ is reflex,
and we can use its label to access
$b_{\id}=\br(\nd(s,t))_{\id}$. The vertex $b$
splits $I(s,t) = [\nd(s,t),\fd(s,t)]$
into two disjoint
subintervals $[\nd(s,t),b]$ and
$[\cv(b),\fd(s,t)]$.
Also, $b$ and $\cv(b)$ are not visible from $s$,
as they are located strictly between the far and the near dominator.
Based on $b_{\id}$, we can now decide on the next hop.

If $t \in [\nd(s,t),b]$, we take the hop to
$\nd(s,t)$. If $t=b$, our packet uses
a shortest path of length $2$. Thus, assume that $t$ lies
between $\nd(s,t)$ and $b$. This is only possible
if $b$ is $\ell$-reflex, and we can apply
\cref{lem:any-path-includes-interval-border}
to see that any shortest path from $s$ to $t$ includes
$\nd(s,t)$ or $b$. But since $d(s,b)=2$, our
data packet routes along a shortest path.

If $t\in[\cv(b),\fd(s,t)]$, we take the hop
to $\fd(s,t)$. If $t=\cv(b)$, our packet
uses a shortest path of length $2$. Thus, assume that $t$
lies between $\cv(b)$ and $\fd(s,t)$. This
is only possible if $\cv(b)$ is $r$-reflex,
so we can apply \cref{lem:any-path-includes-interval-border}
to see that any shortest path from $s$ to $t$ uses
$\fd(s,t)$ or $\cv(b)$. Since
$d(s,\cv(b))=2$,
our packet routes along a shortest path.
The following theorem summarizes our discussion.

\begin{restate}[thm:simple-routing]
\begin{theorem}[restated]
Let $P$ be a simple histogram with $n$ vertices. There is a routing scheme
for $G(P)$ with a routing table with 1 bit, without headers, having label
size $2\cdot\lceil\log n\rceil$, such that we can route between any two
vertices on a shortest path.
\end{theorem}
\end{restate}

\section{Double Histograms}
Let $P$ be a double histogram with $n$ vertices. Similar to the simple
histogram case, we first focus on the
visibility and the structure of shortest paths in $P$.
Again, if a target vertex $t$ is not in the interval $I(s)$ of a current
vertex $s$, we should widen the interval as fast as possible. However, 
in contrast
to simple histograms, we can now change sides arbitrarily often.
Nevertheless, we can guarantee that in each step, the interval comes closer
to $t$. Once we have reached the case that $t$ is in the interval of 
the current vertex, we again
have to find the right pocket. Unlike in simple histograms, this case
is now simpler to describe.

\subsection{Visibility in Double Histograms}
The structure of the shortest paths in double histograms can be much 
more involved than in simple histograms; in particular,
\cref{lem:any-path-includes-interval-border}
does not hold anymore. However, the following
observations provide some structural insight that 
can be used for an efficient routing scheme.

\begin{observation}
\label{obs:mutually-vis-iff}
Two vertices $v,w$ are co-visible
if and only if $v \in I(w)$ and $w \in I(v)$.
\end{observation}

\begin{proof}
The forward direction is immediate, since co-visibility implies
$v \in N(w) \subseteq I(w)$ and $w \in N(v) \subseteq I(v)$. 
For the backward direction, let
$Q$ be the rectangle spanned by $v$ and $w$. Since 
$v \in I(w)$ and $w \in I(v)$,
the upper and lower boundary of $Q$ do
not contain a point outside $P$. As $P$ is a double histogram,
this implies that the left and right boundary of $Q$ also do not
contain any point outside $P$. The claim follows since $P$ has no
holes.
\end{proof}

\begin{observation}
\label{obs:overlapping-intervals}
Let $a$, $b$, $c$, and $d$ be vertices in $P$ with
$a_x\leq b_x\leq c_x \leq d_x$. If $a\in I(c)$ and $d\in I(b)$, then
$b$ and $c$ are co-visible.
\end{observation}
\begin{proof}
This follows immediately from \cref{obs:mutually-vis-iff}.
\end{proof}

\begin{observation}
\label{obs:intervals-on-same-side-are-subsets}
The intervals on one side of $P$ form a laminar family,
i.e., 
for any two vertices $v$ and $w$ on the same side of the
base line, we have (i) $I(v) \cap I(w) = \emptyset$, (ii)
$I(v)\subseteq I(w)$, or (iii) $I(w)\subseteq I(v)$.
\end{observation}
\begin{proof}
Suppose there are two vertices $v$ and $w$ on the same side
of $P$ with
$\ell(v)_x< \ell(w)_x\leq r(v)_x<r(w)_x$. By \cref{obs:overlapping-intervals},
$\ell(w)$ and $r(v)$ are co-visible. Since $\ell(w)$ and $r(v)$ are on
the same side of $P$, either $r(v)$ cannot see any
vertex to the left of $\ell(w)$ or $\ell(w)$ cannot see any vertex to the
right of $r(v)$. This contradicts the fact that the $\ell(v)$ and $r(v)$
as well as $\ell(w)$ and $r(w)$ must be co-visible.
\end{proof}

\subsection{Paths in a Double Histogram}

To understand shortest paths in double histograms, we distinguish
three cases, depending on where $t$ lies relative
to $s$. First, if $t$ is close, i.e., if $t\in I(s)$, we 
focus on the near and
far dominators. Second, if $t\notin I(s)$ but
there is a vertex $v$ visible from $s$ with $t\in I(v)$, 
then we can find a vertex on a shortest path from 
$s$ to $t$. Third, if there is
no visible vertex $v$ from $s$ such that $t\in I(v)$, we can
apply our intuition from simple histograms: go as fast as possible 
towards the
base line. Details follow.

\subparagraph*{The target is close.}
Let $s, t$ be two vertices with
$t \in I(s) \setminus N(s)$. In contrast to simple histograms,
$\fd(s,t)$ now might not be a vertex.
Furthermore, $\fd(s,t)$ and 
$\nd(s,t)$ might be on different sides of the 
base line. 
In this case,
\cref{lem:nd-sees-fd-simple} no longer holds. 
However,
the next lemma establishes a visibility
relation between them; see \cref{fig:fd-and-nd}.
\begin{figure}
\begin{center}
\includegraphics[scale=1.0]{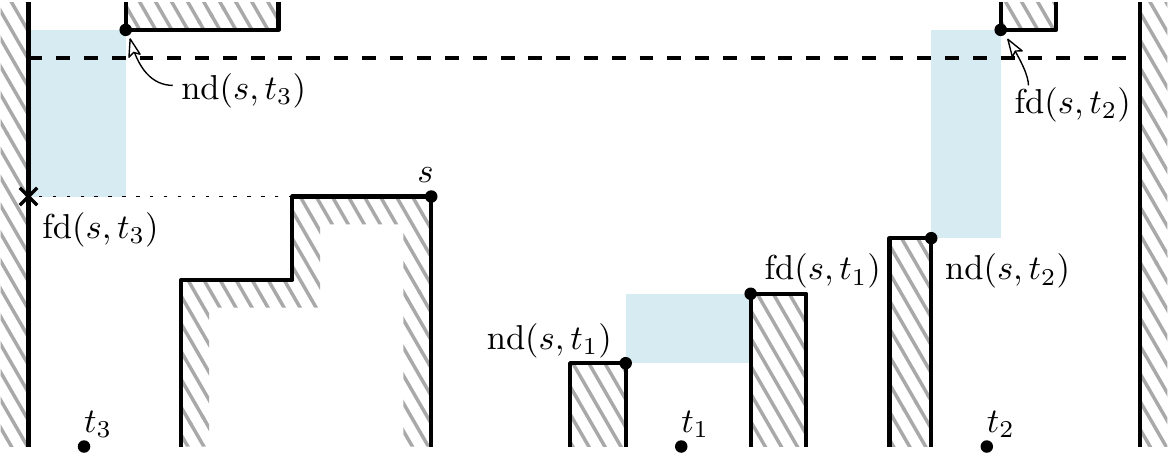}
\end{center}
\caption{The far and the near dominator can see each other.}
\label{fig:fd-and-nd}
\end{figure}
\begin{lemma}
\label[lemma]{lem:nd-sees-fd}
Let $s, t \in V(P)$ with
$t \in I(s) \setminus N(s)$.
Then, $\nd(s,t)$ and
$\fd(s,t)$ are co-visible.
\end{lemma}
\begin{proof}
Without loss of generality, $t$ is strictly to the 
right of $s$. Suppose for a contradiction that
$r(\nd(s,t))$ is strictly left of $\fd(s,t)$. Then, 
we get $r(\nd(s,t))\in I(s)$. 
Also, $s \in I(\nd(s,t))
\subseteq I(r(\nd(s,t)))$. Hence, by \cref{obs:mutually-vis-iff},
$s$ can see $r(\nd(s,t))$. But then $r(\nd(s,t))$
is a vertex strictly between the near and far dominator
visible from $s$, 
contradicting the choice of the dominators.
Thus, $s_x \leq \nd(s,t)_x\leq\fd(s,t)_x\leq
r(\nd(s,t))_x$, and
\cref{obs:overlapping-intervals} gives the result.
\end{proof}
The proof of the next lemma uses \cref{lem:nd-sees-fd} 
to find a shortest path vertex.
\begin{lemma}
\label[lemma]{lem:shortest-path-via-dn-or-df}
One of $\nd(s,t)$ or $\fd(s,t)$ is on a
shortest path from $s$ to $t$.
If $\fd(s,t)$ is not a vertex, then
$\nd(s,t)$ is on a shortest path from $s$ to $t$.
\end{lemma}

\begin{proof}
Without loss of generality, $t$ is to the right of $s$.
Let $\pi: \langle s = p_0, \dots, p_k = t\rangle$ 
be a shortest path from $s$ to $t$,
and let $p_j$ be the last vertex outside of $I(s,t)$.
If $j = 0$, then $p_{j+1}$ must be one of the 
dominators, since by definition they are the only 
vertices in $I(s, t)$ visible from $s$.
Now, assume $j \geq 1$.
If $p_j$ is to the left of $\nd(s,t)$, we apply
\cref{lem:nd-sees-fd} and \cref{obs:overlapping-intervals} on the
four points $p_j$, $\nd(s,t)$, $p_{j+1}$, and
$\fd(s,t)$ to conclude that $\nd(s,t)$ can see
$p_{j+1}$. Symmetrically, if $p_j$ is to the right of $\fd(s,t)$,
the same argument shows that the far dominator
can see $p_{j+1}$.
Thus, depending on the position of $p_j$ 
we can exchange the subpath $p_1, \dots, p_j$ in $\pi$ by 
$\nd(s,t)$ or $\fd(s, t)$
and get a valid path of length
$k - j + 1 \leq k$. The second part of the lemma holds because 
$p_j$ cannot be to the right of $\fd(s,t)$, 
if $\fd(s,t)$ is not a vertex but a point on the right boundary.
\end{proof}

Next, we consider the case where $\fd(s,t)$ is a vertex
but not on a shortest path from $s$ to $t$. Then, $\fd(s,t)$
cannot see $t$, and we define
$\fd^2(s,t)=\fd(\fd(s,t),t)$. By \cref{lem:nd-sees-fd},
$\nd(s,t)$ and $\fd(s,t)$ are co-visible,
so $\fd^2(s, t)$ has
to be in the interval $[\nd(s,t),t]$, and therefore it is
a vertex. The following lemma states that
$\fd^2(s,t)$ is strictly closer to $t$ than $s$; 
see \cref{fig:fd-fd}.
\begin{figure}
\begin{center}
\includegraphics[scale=1.0]{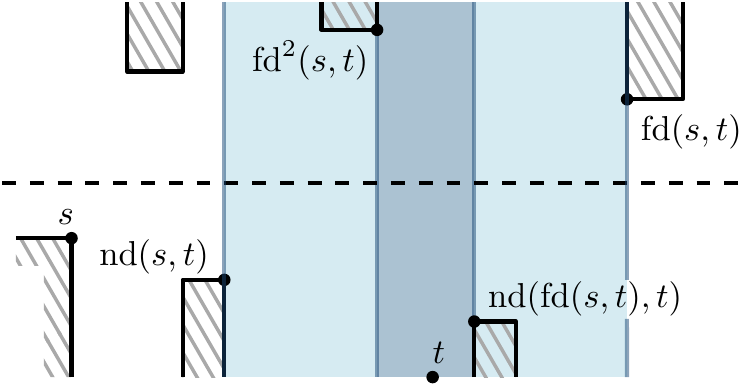}
\end{center}
\caption{$\fd^2(s,t)$ lies between $\nd(s,t)$ and
$\fd(s,t)$ and is closer to $t$ than $s$. The darker region
is $I(\fd(s,t),t)$ and a subset of $I(s,t)$, the brighter region.}
\label{fig:fd-fd}
\end{figure}
\begin{lemma}
\label[lemma]{lem:fd-fd-is-closer}
If $\fd(s,t)$ is a vertex but not
on a shortest path from $s$ to $t$, then we have
$d(\fd^2(s,t),t) = d(s,t)-1$.
\end{lemma}

\begin{proof}
Without loss of generality, $t$ is to the right of $s$.
By \cref{lem:shortest-path-via-dn-or-df}, $\nd(s,t)$ lies
on a shortest path from $s$ to $t$.
Let $\langle s = p_0, \nd(s, t) = p_1, p_2, \dots, p_k = t\rangle$ 
be such a shortest path.
We claim that $\fd^2(s,t)$ can see $p_2$.
Then, $\langle\fd^2(s,t), p_2, \dots, p_k =  t\rangle$ is a
valid path of length $k - 1 = d(s, t)-1$.
To prove that $\fd^2(s, t)$ can indeed see $p_2$, 
we show that
$p_2\in I\big(\!\fd^2(s,t)\big)$ and $\fd^2(s,t)\in I(p_2)$ and
then apply \cref{obs:mutually-vis-iff}.

First, we show $p_2 \in I(\fd(s,t),t)$ by contradiction. Thus,
suppose $p_2 \notin I(\fd(s,t),t)$.
Since $t \in I(\fd(s,t),t)$, there is a $j\geq 2$ 
with $p_{j+1}\in I(\fd(s,t),t)$
and $p_j\notin I(\fd(s,t),t)$.
First, if $p_{j,x}<\fd^2(s,t)_x$, then
$p_{j,x}<\fd^2(s,t)_x\leq p_{j+1,x}\leq \nd(\fd(s,t),t)_x$.
By \cref{lem:nd-sees-fd}, $\fd^2(s,t)$ and $\nd\big(\!\fd(s,t), t\big)$
are co-visible, so \cref{obs:overlapping-intervals} implies that
$\fd^2(s,t)$ and $p_{j+1}$ are co-visible.
Then $\langle s,\fd(s,t),\fd^2(s,t),p_{j+1},\dots,p_k = t\rangle$ is
a valid path of length $k-j+2\leq k$, contradicting the assumption 
that $\fd(s,t)$ is not on a shortest path.
If $\nd(\fd(s,t),t)_x<p_{j,x}$, it follows with the same reasoning that
$\nd(\fd(s,t),t)$ and $p_{j+1}$ are co-visible then
$\fd(s,t)$ is on
$\langle s, \fd(s,t),\nd(\fd(s,t),t),p_{j+1},\dots,p_k = t\rangle$
which is a valid path of length $k-j+2\leq k$.
This again contradicts the assumption.

Now, since $I(\fd(s,t),t) = \big[\!\fd^2(s,t),\nd(\fd(s,t),t)\big]
\subseteq I(\fd^2(s,t))$, 
we get
$p_2 \in I(\fd^2(s,t))$.
Since $p_2$ sees $\nd(s,t)$ which is to the left of
$\fd^2(s,t)$ and since $p_2$ is in $I(\fd(s,t),t)$, and
thus to the right of $\fd^2(s,t)$,
it follows that $\fd^2(s,t)\in I(p_2)$.
\end{proof}

\subparagraph*{The target can be made close in one step.}
Let $s, t$ be two vertices with $t \notin I(s)$ but there is a vertex
$v \in N(s)$ with $t\in I(v)$.
For clarity of presentation, we will always assume that $s$ 
is below the base line.
The crux of this case is this: there might
be many vertices visible
from $s$ that have $t$ in their interval. However, we can find a best
vertex as follows: once $t$ is in the interval
of a vertex, the goal is to shrink the interval as fast as possible.
Therefore, we must find a vertex $v \in N(s)$ whose left or right
interval boundary is closest to $t$ among all vertices in $N(s)$.
This leads to the following inductive definition of two sequences
$a^i(s)$ and $b^i(s)$ of vertices in $N(s)$.
For $i = 0$, we let $a^0(s) = b^0(s) = s$.
For $i > 0$,
if the set
$A^i(s)=\{v\in N(s) \mid \ell(v)_x < \ell(a^{i-1}(s))_x\}$
is nonempty, we define $a^i(s) = \argmin\{v_x \mid v\in A^i(s)\}$;
and $a^i(s)=a^{i-1}(s)$, otherwise. If the set
$B^i(s)=\{v \in N(s) \mid r(v)_x > r(b^{i-1}(s))_x\}$ is nonempty, we 
define $b^i(s) = \argmax\{v_x \mid v\in B^i(s)\}$;
and $b^i(s)=b^{i-1}(s)$, otherwise. We force unambiguity by choosing the
vertex closer to the base line.
Let $a^*(s)$ be the vertex with $a^*(s)=a^i(s)=a^{i-1}(s)$, for 
an $i > 0$, and 
$b^*(s)$ the vertex with $b^*(s)=b^i(s)=b^{i-1}(s)$,
for an $i > 0$. If the context is clear, we write
$a^i$ instead of $a^i(s)$ and $b^i$ instead of $b^i(s)$.

Let us try to understand this definition. 
For $i  \geq 0$, we write $\ell^i$ for 
$\ell(a^i)$; and we write $\ell^*$ for $\ell(a^*)$.
Then, we have $a^0 = s$ and $\ell^0 = \ell(s)$.
Now, if $\ell(s)$ is not a
vertex, then $a^* = s$, because there is no 
vertex whose left point is strictly to the left of the left 
boundary of $P$.
On the other hand, if $\ell^0$ is a vertex in $P$, we have 
$a^1 = \ell^0 = \ell(s)$, and $[\ell^1, a^1]$ is an 
interval between points on the lower side of $P$.
Then comes a (possibly empty) sequence of 
intervals $[\ell^2, a^2], [\ell^3, a^3], \dots, [\ell^k, a^k]$
between points on the upper side of $P$;
possibly followed by the interval $[\ell(r(s)), r(s)]$.
There are four possibilities for $a^*$: it could 
be $s$, $\ell(s)$, a vertex $a^i$ on the upper side of $P$,
or $r(s)$. If $a^* \neq s$, then the intervals
$[\ell^1, a^1] \subset [\ell^2, a^2] \subset \dots \subset [\ell^*, a^*]$
are strictly increasing: $\ell^i$ is strictly to the left
of $\ell^{i-1}$ and $a^i$ is strictly to the right of $a^{i-1}$; 
see \cref{fig:a-i-b-i}.
Symmetric observations apply for the $b^i$; we write 
$r^i$ for $r(b^i)$ and $r^*$ for $r(b^*)$. 

\begin{figure}
\begin{center}
\includegraphics[scale=0.9]{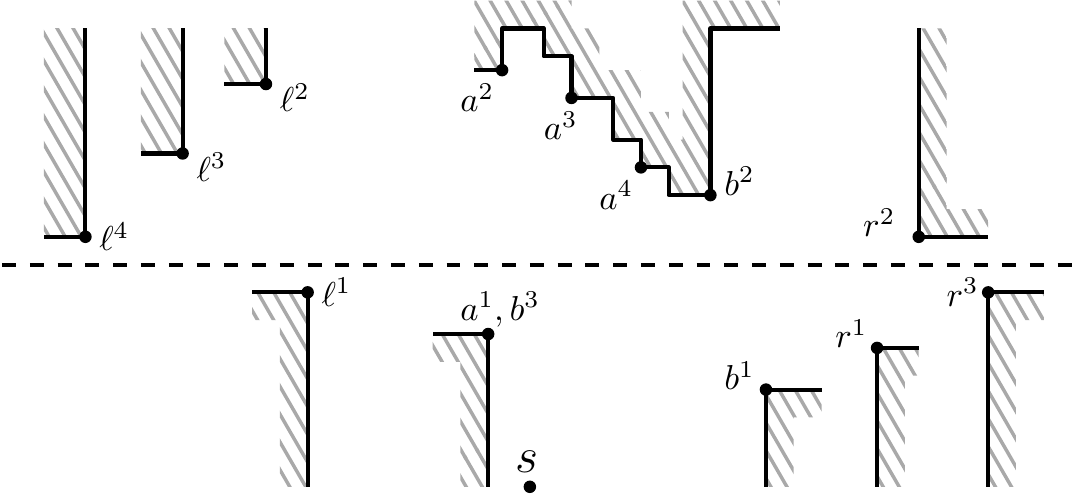}
\end{center}
\caption{The vertices $a^i$ and $b^i$ are illustrated. Observe that
$\ell(s)=a^1=b^3$ and $r(s)=b^1$.}
\label{fig:a-i-b-i}
\end{figure}
\begin{lemma}
\label[lemma]{lem:ai-seeslai-1}
For $i \geq 1$, the vertices
$\ell^{i-1}$ and  $a^{i}$ as well as 
$r^{i-1}$ and $b^i$ are 
co-visible.
\end{lemma}
\begin{proof}
We focus on $\ell^{i-1}$ and $a^i$. 
We show that $\ell^{i-1} \in I(a^i)$ and 
$a^i \in I(\ell^{i-1})$; the lemma 
follows from \cref{obs:mutually-vis-iff}.
The claim
$\ell^{i-1} \in I(a^i)$ is due to the facts that
$[\ell^i, a^i] \subseteq I(a^i)$ and $\ell^{i-1} \in [\ell^i, a^i]$
(this holds also for $i = 1$, as then $\ell^{i-1} = a^i$).
Next, since $a^{i-1} \in I(\ell^{i-1})$, the vertex $a^{i-1}$ is 
to the left of $r(\ell^{i-1})$; and
since $\ell^{i-1} \in I(r(\ell^{i-1}))$,
the point $\ell(r(\ell^{i-1}))$ is to
the left of $\ell^{i-1}$.
Thus, if $r(\ell^{i-1})$ is visible from $s$,
we have $a^i = r(\ell^{i-1})$, by the definition of $a^i$.
On the other hand, if $r(\ell^{i-1})$ is not visible from $s$,
the visibility must be blocked by $r(s)$, and then $a^i = r(s)$.
In either case, we have
$a^i \in [a^{i-1},r(\ell^{i-1})] \subseteq I(\ell^{i-1})$, 
as desired.
\end{proof}
Finally, the next lemma tells us the following: if $t \in [\ell^ *,r^*]$
we find a vertex $v \in N(s)$ with $t \in I(v)$. Its quite technical
proof needs \cref{lem:ai-seeslai-1}.
\begin{lemma}
\label[lemma]{lem:t-in-extended-interval}
If $t\in[\ell^i,\ell^{i-1}]$, for some $i \geq 1$,
then $a^i$ is on a shortest path from $s$ to $t$.
If $t\in[r^{i-1},r^i]$, for some $i \geq 1$,
then $b^i$ is on a shortest path from $s$ to $t$.
\end{lemma}
\begin{proof}
We focus on the first statement; see
\cref{fig:t-in-extended-interval}.
Let $\pi: \langle s = p_0,\dots,p_k = t \rangle$ be a 
shortest path from $s$ to $t$, and let $p_j$ be 
the last vertex on $\pi$ outside 
of $[\ell^i,\ell^{i-1}]$. 
If $j = 0$, then $p_{j+1}$ must be $\ell^0 = \ell(s)$, because this 
is the only vertex $\ell^{i-1}$ visible from $s$.
Then, $i = 1$ and $a^i = \ell(s)$ is on $\pi$. 
From now on, we assume that $j \geq 1$.

\begin{figure}
\begin{center}
\includegraphics[scale=0.8]{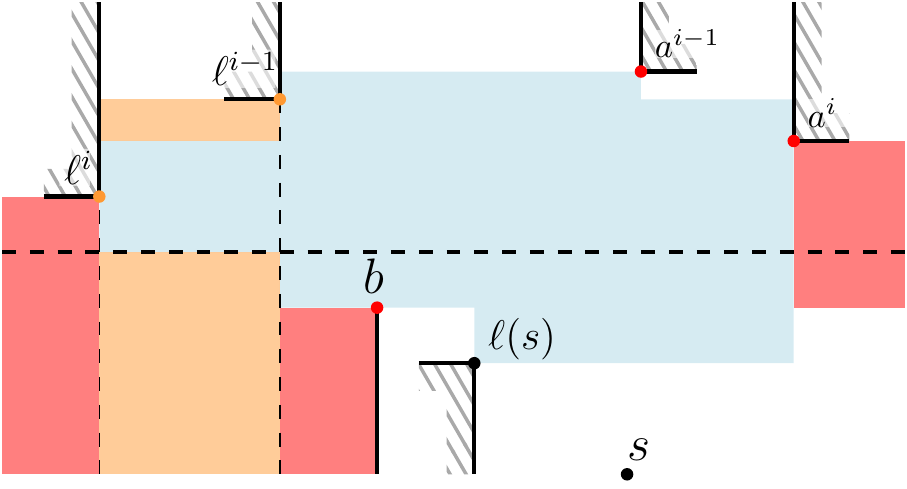}
\end{center}
\caption{The vertex $a^i$ is on a shortest path. The vertex $p_j$ can
lie in the red regions, the vertex $p_{j+1}$ can lie in the orange region,
and the blue region cannot contain any point outside of $P$.}
\label{fig:t-in-extended-interval}
\end{figure}

First, suppose that
$p_{j+1}$ and $a^i$ are co-visible. Then
$\langle s, a^i, p_{j+1},\dots, p_k\rangle$ is a path from $s$ 
to $t$ that uses $a^i$ and has length $k-j + 1 \leq k$.
Second, suppose that $p_{j+1}$ and $a^i$ are not 
co-visible. Then, the contrapositive of 
\cref{obs:overlapping-intervals} applied to the four points 
$\ell^i$, $p_{j+1}$, $a^i$, and $p_j$ shows that $p_j$ is 
strictly to the left of $a^i$. There are two subcases, depending on whether 
$p_j$ is strictly to the left of $\ell^i$ or strictly 
to the right of $\ell^{i-1}$.

If $p_j$ is strictly to the left of $\ell^i$, then $j \geq 2$, since
$\ell^i$ is to the left of $\ell^0 = \ell(s)$ and we need at least two 
hops to reach a point strictly to the left of $\ell(s)$ from $s$. We apply
\cref{obs:overlapping-intervals} on the four points $p_j$, $\ell^i$,
$p_{j+1}$, and $a^i$, and get that $\ell^i$ and $p_{j+1}$ 
are co-visible. Hence,
$\langle s,a^i,\ell^i,p_{j+1},\dots,p_k\rangle$ is a path 
that uses $a^i$ and has length
$k-j+2\leq k$.

Finally, assume that $p_j$ is strictly to the right of $\ell^{i-1}$.
By \cref{lem:ai-seeslai-1}, $a^i$
can see $\ell^{i-1}$. Thus, $p_{j+1} \neq \ell^{i-1}$ and
there is no vertex strictly between $\ell^{i-1}$ and $a^i$ 
on the same side as 
$\ell^{i-1}$ that can see a
vertex strictly to the left of $\ell^{i-1}$.
Thus, $p_j$ and $a^{i-1}$ are
on different sides of the base line. 
Let $b$ be the rightmost vertex that (i) lies on the 
same side of $P$ as $p_j$; (ii) is strictly between 
$\ell^{i-1}$ and $a^i$; (iii) is closest to the base line.
The vertex $b$ exists (since $p^j$ is a candidate),
is not visible from  $s$
(because $b$ is strictly
left of $a^i$ and can see strictly left of $\ell^{i-1}$);
and thus strictly left of $\ell(s)$.
The vertex $p_j$
cannot be strictly to the right of $b$, as otherwise $b$ would 
obstruct visiblity between $p_j$ and
$p_{j+1}$. 
We conclude that $j \geq 2$,  since we need at least two hops 
to reach a point strictly to the left of $\ell(s)$ from $s$.
If $p_j \in \{b,\cv(b)\}$, $a^i$ can see $p_j$ and thus,
$\langle s,a^i,p_j,p_{j+1},\dots,p_k\rangle$ is a path of length $k-j+2\leq k$
using $a^i$.
If $p_j \not\in \{b,\cv(b)\}$, then
$b$ is strictly closer to the base line than $p_j$.
Then, we have $j\geq 3$, because we need two hops to cross the 
vertical line through $\ell(s)$ and one more hop to cross the
horizontal line through $b$.
We apply \cref{obs:overlapping-intervals} on 
the four points $p_{j+1}$, $\ell^{i-1}$, $p_j$, and $b$ to 
conclude that $\ell^{i-1}$ can see $p_j$. Hence, 
$\langle s,a^i,\ell^{i-1},p_j,p_{j+1},\dots,p_k\rangle$ is a path of
length $k-j+3\leq k$ using $a^i$. 
\end{proof}

\subparagraph*{The target is far away.}
Finally, we consider the case that there is no vertex
$v \in N(s)$ with $t\in I(v)$, i.e., $t\notin[\ell^*,r^*]$.
The intuition now is as follows: to widen the interval, we should go to
a vertex that is visible from $s$, but closest to the base line. In simple
histograms, there was only one such vertex, but in double histograms 
there might
be a second one on the other side. These two vertices are the
\emph{dominators of $s$}. These two dominators might have their
own dominators, and so on. This leads to the following inductive definition.

For $k \geq 0$, 
we define the $k$-th \emph{bottom
dominator} $\bd^k(s)$, 
the $k$-th \emph{top dominator} $\td^k(s)$, 
and the $k$-th \emph{interval} $I^k(s)$ of $s$.
For any set $ Q\subset V(P)$, we write  $Q^-$ (resp.~$Q^+$) for all
points in $Q$ below (resp.~above) the base line.
We set $\bd^0(s) = \td^0(s) = s$ and
$I^0(s) = \{ s \}$.
For $k > 0$, we set 
$I^k(s) = I(\bd^{k-1}(s))\cup I(\td^{k-1}(s))$.
If $I^k(s)^-$
is nonempty, we let $\bd^k(s)$ be the leftmost 
vertex inside $I^k(s)^-$ that minimizes the distance to the base line. 
If $I^k(s)^+$ is nonempty, we let $\td^k(v)$ be the 
leftmost vertex inside
$I^k(s)^+$ that minimizes the distance to the base line; see \cref{fig:dominators}. 
If one of the
two sets is empty, the other one has to be nonempty, since
$s \in I^k(s)$. In this case, we let
$\td^k(s) = \bd^k(s)$. 
We write $\bd(s)$ for $\bd^1(s)$ and $\td(s)$ for $\td^1(s)$.
\begin{figure}
\begin{center}
\includegraphics[scale=1]{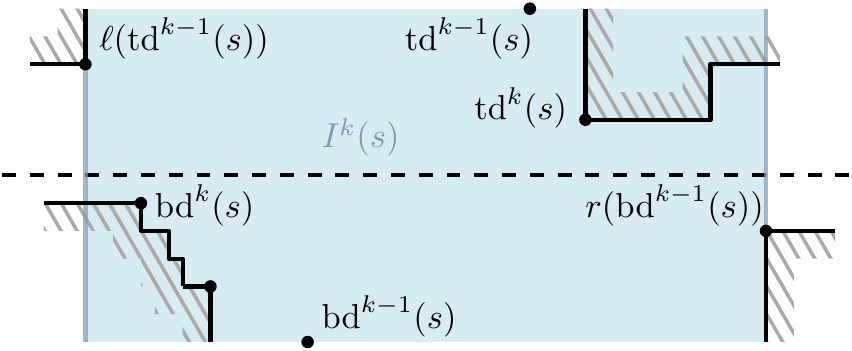}
\end{center}
\caption{The $k-1$-th and $k$-th dominators and the $k$-th interval.}
\label{fig:dominators}
\end{figure}
Observe, that $I^1(s) = I(s)$ and $I^2(s) = [\ell^*,r^*]$.
If $I(\bd^{k-1}(s)) = V(P)$, we have 
$\bd^k(s)= \bd^{k-1}(s)$. 
The same holds for the top dominator.
We provide a few technical properties concerning the
$k$-th interval as well as the $k$-th dominators.

\begin{lemma}
\label[lemma]{lem:Ik-inside-Ibdk-and-Itdk}
For any $s \in V(P)$ and $k \geq 0$,
we have
$I^k(s) \subseteq I(\bd^k(s))\cap
I(\td^k(s))$.
\end{lemma}
\begin{proof}
We have $I^k(s) \subseteq I(\bd^k(s))$,
since by definition the interval $I^k(s)$ contains no 
vertex that is on the same side as $\bd^k(s)$ and strictly
closer to the base line, so no vertex can obstruct horizontal 
visibility of $\bd^k(s)$ in $I^k(s)$.
Analogously, we have $I^k(s)\subseteq I(\td^k(s))$. 
The claim follows.
\end{proof}

\begin{lemma}
\label[lemma]{lem:dominators-see-each-other}
For any $s \in V(P)$ and $k \geq 0$,
$\bd^k(s)$ and $\td^k(s)$ are co-visible.
\end{lemma}
\begin{proof}
By definition and \cref{lem:Ik-inside-Ibdk-and-Itdk}, 
$\td^k(s) \in I^k(s) \subseteq I(\bd^k(s))$ and
$\bd^k(s) \in I^k(s) \subseteq I(\td^k(s))$.
The claim now follows from \cref{obs:mutually-vis-iff}.
\end{proof}
The following lemma seems rather specific but it will be needed
later on to deal with short paths. 
\begin{lemma}
\label[lemma]{lem:I3-equals-exended-of-dominators}
For any $s \in V(P)$, we have
$I^3(s)=I^2\big(\!\bd(s)\big) \cup I^2\big(\!\td(s)\big)$.
\end{lemma} 

\begin{proof}
We begin by showing that
\begin{equation}
\label{equ:bdtwo}
I\big(\!\bd^2(s)\big) = I\big(\!\bd(\td(s))\big) \cup I\big(\!\bd(\bd(s))\big).
\end{equation}
If $\bd(s)$ is above the base line, 
then $\bd(s) = \td(s)$
and $I^2(s) = I(\bd(s)) \cup I(\td(s)) = I(\td(s))$. The 
definition of $\bd^2(s)$ then gives 
$\bd^2(s) = \bd(\td(s))$, and
(\ref{equ:bdtwo}) follows.

If $\bd(s)$ is below the base line,
the vertex $b_1 = \bd(\bd(s))$ is below the base line.
Let $b_2 = \bd(\td(s))$. By
\cref{lem:dominators-see-each-other},
$\bd(s)$ and $\td(s)$ are co-visible,
so $\bd(s) \in I(\td(s))^-$. Therefore, $b_2$
is below the base line. Since $I(b_1)$ and $I(b_2)$
are not disjoint (both contain $s$) and since 
$b_1$ and $b_2$ are on the same
side of the base line,
\cref{obs:intervals-on-same-side-are-subsets} gives
$I(b_1) \subseteq I(b_2)$ or
$I(b_2) \subseteq I(b_1)$. Because $\bd^2(s)$ is the highest vertex
in $\big(I(\bd(s)) \cup I(\td(s)) \big)^-$, we get that $\bd^2(s)$ is
$b_1$ or $b_2$, and (\ref{equ:bdtwo}) follows also in this 
case.
Symmetrically, we have
\begin{equation}
\label{equ:tdtwo}
I(\td^2(s))=I(\td(\td(s))) \cup I(\td(\bd(s))).
\end{equation}
We use the definitions and (\ref{equ:bdtwo},\ref{equ:tdtwo}) to get

\begin{multline*}
I^3(s) = I\big(\!\bd^2(s)\big) \cup I\big(\!\td^2(s)\big) \\
 = I\big(\!\bd(\td(s))\big) \cup
I\big(\!\bd(\bd(s))\big) \cup
I\big(\!\td(\td(s))\big) \cup
I\big(\!\td(\bd(s))\big) 
= I^2\big(\!\bd(s)\big) \cup I^2\big(\!\td(s)\big),
\end{multline*}

as desired.
\end{proof}

Intuitively,
the meaning of $I^k(s)$ is as follows: let $\ell$ be the leftmost and
$r$ be the rightmost vertex with hop distance exactly $k$ from $s$,
then, $I^k(s) = [\ell,r]$. We do not really need 
this property. So we leave it as an exercise for the reader to find a 
proof for this. Instead, we prove the following weaker statement.
For this, recall that due to its definition, $\bd^k(s)$ might not be
on the lower side of the histogram (and
$\td^k(s)$ might not be on the upper side).

\begin{lemma}
\label[lemma]{lem:distance-k-implies-in-Ik}
Let $k \geq 0$ and let
$s, t \in V(P)$ with $d(s, t) \leq k$. 
Then, $t \in I^k(s)$. 
\end{lemma}

\begin{proof}
We show that for any $j \geq 0$ and any vertex
$v\in I^j(s)$, we have
$N(v) \subseteq I^{j+1}(s)$. The lemma then follows
by induction.
If $v \in I^j(s)^-$, then $\bd^j(s)$
is on the lower side, and by definition, 
$I(v)\subseteq I(\bd^j(s))$. If $v \in I^j(s)^+$, 
by a similar argument $I(v)\subseteq I(\td^j(s))$.
Thus,
$N(v)\subseteq I(v)\subseteq I(\bd^j(s)) \cup I(\td^j(s)) = I^{j+1}(s)$,
as desired.
\end{proof}

Let $k \geq 0$ and $s \in V(P)$.
For $i = 1, \dots, k$, by \cref{lem:Ik-inside-Ibdk-and-Itdk}, 
$\bd^{i-1}(s), \td^{i-1}(s) \in I(\bd^i(s)) \cap
I(\td^i(s))$.
Moreover, by definition,
$\bd^i(s), \td^i(s) \in I(\bd^{i-1}(s))
\cup I(\td^{i-1}(s))$. \cref{obs:mutually-vis-iff} now
says that both $\bd^i(s)$ and $\td^i(s)$ 
can see at least one of $\bd^{i-1}(s)$ or $\td^{i-1}(s)$.
Therefore, there is a path 
$\pi_b(s, k): \langle s = p_0, \dots, p_k = \bd^k(s)\rangle$ 
from $s$ to $\bd^k(s)$
and a path
$\pi_t(s, k): \langle s = q_0, \dots, q_k = \td^k(s)\rangle$
from $s$ to $\td^k(s)$ with $p_i, q_i \in  \{\bd^i(s), \td^i(s)\}$,
for $i = 0, \dots, k$.
We call $\pi_b(s, k)$ and $\pi_t(s, k)$ the \emph{canonical path} 
from $s$ to $\bd^k(s)$ and from $s$ to  $\td^k(s)$, 
respectively.
The following two lemmas show that for every $t\notin I^ {k+1}(s)$
one of the canonical paths is the prefix of a shortest path from $s$
to $t$. To show
\cref{lem:t-outside-kpp-interval} we need
\cref{lem:distance-k-implies-in-Ik} as well as
\cref{lem:dominators-have-distance-k}.
\begin{lemma}
\label[lemma]{lem:dominators-have-distance-k}
Let $k \geq 1$ and $s \in V(P)$. If
$I(\bd^{k-1}(s)) \neq V(P)$ we have
$d(s,\bd^k(s)) = k$.
If $I(\td^{k-1}(s)) \neq V(P)$ we have
$d(s,\td^k(s)) = k$.
\end{lemma}
\begin{proof}
On the one hand,
$d(s,\bd^k(s)) \leq |\pi_b(s,k)| = k$ and
$d(s,\td^k(s)) \leq |\pi_t(s,k)| = k$.
On the other hand, we show that
$\bd^k(s) \notin I^{k-1}(s)$ and
$\td^k(s) \notin I^{k-1}(s)$. The claim then follows from
the contrapositive of
\cref{lem:distance-k-implies-in-Ik}.

\textbf{Case 1:} First, assume that $\bd^{k-1}(s) \in I^{k-1}(s)^-$. 
Since $I(\bd^{k-1}(s)) \neq P$, at least one of its bounding
points is a vertex $v$ contained in $I^k(s)$. Then, $v$ 
is strictly closer to the base line than $\bd^{k-1}(s)$, and since 
$v$ is a candidate for $\bd^k(s)$, the same applies to
$\bd^k(s)$.  It follows that
$\bd^k(s) \not\in I^{k-1}(s)$. Similarly,
we get that if $\td^{k-1}(s) \in I^{k-1}(s)^+$, the vertex
$\td^k(s)$ is not in $I^{k-1}(s)$.

\textbf{Case 2:} Second, assume that $\bd^{k-1}(s)\in I^{k-1}(s)^+$.
Then, $\bd^k(s) \notin I^{k-1}(s)^-$, since this set is empty.
Thus, suppose for a contradiction that 
$\bd^k(s) \in I^{k-1}(s)^+$. This can
only be the case if $\bd^{k-1}(s)=\td^{k-1}(s)$ and
$\bd^k(s) = \td^k(s)$. However, in 
Case~1 we showed that $\td^k(s) \not\in I^{k-1}(s)^+$ if
$\td^{k-1}(s) \in I^{k-1}(s)^+$. Hence,
$\bd^k(s) \notin I^{k-1}(s)^+$, as desired.
\end{proof}

\begin{lemma}
\label[lemma]{lem:t-outside-kpp-interval}
Let $s$ and $t$ be vertices and $k\geq 1$ an integer
such that $t\notin I^{k+1}(s)$. Then $\bd^k(s)$
or $\td^k(s)$ is on a shortest path from $s$ to $t$.
\end{lemma}
\begin{proof}
First, observe that $I^{k+1}(s) = I(\bd^k(s))\cup
I(\td^k(s))\neq P$, as $t\notin I^{k+1}(s)$.
Let $\pi:  \langle s = p_0,\dots,p_m = t\rangle$ 
be a shortest path from $s$ to $t$, 
and $p_j$ the last vertex in $I^{k+1}(s)$. Without
loss of generality, $p_{j+1}$ is strictly to the right of $s$. 
By \cref{lem:distance-k-implies-in-Ik}, we get that $p_j$ 
is not in $I^k(s)$ and thus, again by
\cref{lem:distance-k-implies-in-Ik}, we have
$j\geq d(s,p_j)\geq k+1$.

First, suppose that $p_j$ and $\bd^k(s)$ (resp.
$\td^k(s)$) are co-visible. Then, 
$\pi_b(s,k) \circ \langle p_j, \dots, t \rangle$ (resp.
$\pi_t(s,k) \circ \langle p_j, \dots, t \rangle$)
is a valid path of length $k+1+(m-j)\leq m$. Here, $\circ$
concatenates two paths.
Second, suppose $p_j$ be visible from neither $\bd^k(s)$ nor
$\td^k(s)$. First, we claim that $p_j$ is strictly to the 
right of $\bd^k(s)$ and $\td^k(s)$. 
Otherwise, since $p_{j+1}$ is strictly to the right of both dominators,
we would get $\bd^k(s),\td^k(s)\in I(p_j)$. Moreover,
$p_j\in I^{k+1}(s)=I(\bd^k(s))\cup
I(\td^k(s))$. \cref{obs:mutually-vis-iff} now would imply
that $p_j$ can see $\bd^k(s)$ or $\td^k(s)$---a contradiction. The 
claim follows. Next, we claim $\bd^k(s) \neq \td^k(s)$. If not,
$p_j \in I(\bd^k(s)) = I^{k+1}(s)$, and since $p_j$ can see a point
outside of $I^{k+1}(s)$, we would get 
$p_j \in \{\ell(\bd^k(s)),r(\bd^k(s))\}$,
which again contradicts our assumption that $p_j$ cannot see
$\bd^k(s)$. The claim follows.
There are two cases, depending on which dominator sees
further to the right.

\begin{figure}
\begin{center}
\includegraphics[scale=1]{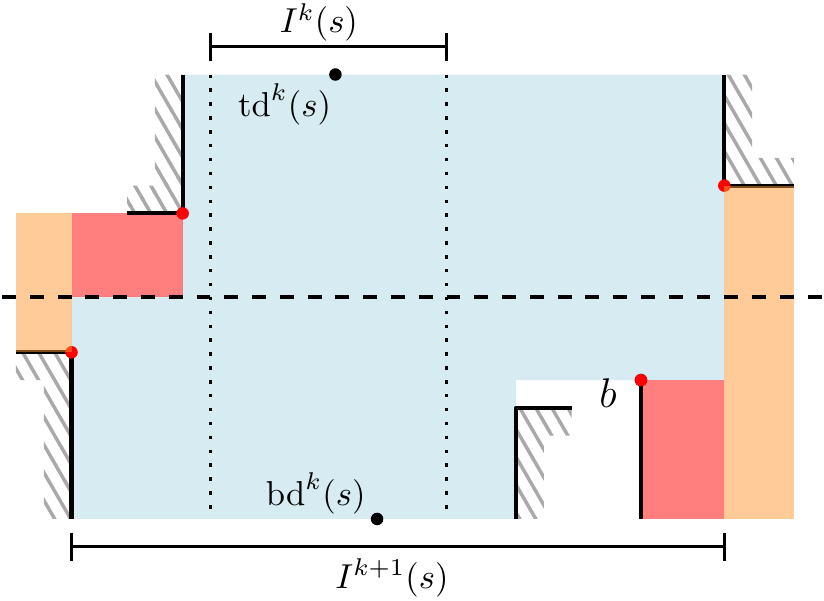}
\end{center}
\caption{$\bd^k(s)$ or $\td^k(s)$ is on a shortest
path. The vertex $p_j$ lies in one of the red regions, $p_{j+1}$
lies in one of the orange regions,
and the blue region cannot contain any point outside of $P$.}
\label{fig:dominators-on-shortest-path}
\end{figure}

\textbf{Case 1:} $r(\bd^k(s))_x<r(\td^k(s))_x$;
see \cref{fig:dominators-on-shortest-path}.
Let $b$ be the
leftmost vertex in $[r(\bd^k(s)),r(\td^k(s))]^-$
closest to the base line. Observe that $b$ is strictly 
to the right of $I^k(s)$, because $r(\bd^k(s))$ is strictly to the 
right of $I^k(s)$.
Since $p_j$ is not visible from $\td^k(s)$, it has to be strictly 
to the right of and strictly below $b$. Next, we claim that 
no vertex $v \in I^k(s)$ can see $p_j$. If one could, by \cref{obs:mutually-vis-iff}, 
we would have $v\in I(p_j)$. But since $p_j$ is strictly to the 
right of and strictly below $b$, then $v$ would be to
the right of $b$, which is impossible. This shows the claim.
Thus, by \cref{lem:distance-k-implies-in-Ik}, $j \geq d(s, p_j) \geq k+2$.
We apply \cref{obs:overlapping-intervals} to $\td^k(s)$, 
$p_j$, $r(\td^k(s))$ and $p_{j+1}$  and get that $r(\td^k(s))$
can see $p_j$. Therefore, $\pi_t(s,k) \circ \langle r(\td^k(s)),p_j,\dots,t\rangle$
is a valid path of length $k+1+(1+m-j)\leq m$.

\textbf{Case 2:} $r(\td^k(s))_x<r(\bd^k(s))_x$.
Let $b$ be the
leftmost vertex in $[r(\td^k(s)),r(\bd^k(s))]^+$
closest to the base line. Observe that $b$ is strictly 
to the right of $I^k(s)$, because $r(\td^k(s))$ is strictly to the 
right of $I^k(s)$.
Since $p_j$ is not visible from $\bd^k(s)$, it has to be strictly 
to the right of and strictly above $b$. Next, we claim that 
no vertex $v \in I^k(s)$ can see $p_j$. If one could, by \cref{obs:mutually-vis-iff}, 
we would have $v\in I(p_j)$. But since $p_j$ is strictly to the 
right of and strictly above $b$, then $v$ would be to
the right of $b$, which is impossible. This shows the claim.
Thus, by \cref{lem:distance-k-implies-in-Ik}, $j \geq d(s, p_j) \geq k+2$.
We apply \cref{obs:overlapping-intervals} to $\bd^k(s)$, 
$p_j$, $r(\bd^k(s))$ and $p_{j+1}$  and get that $r(\bd^k(s))$
can see $p_j$. Therefore, $\pi_b(s,k) \circ \langle r(\bd^k(s)),p_j,\dots,t\rangle$
is a valid path of length $k+1+(1+m-j)\leq m$.
\end{proof}

\subsection{Routing Scheme}

\subparagraph*{Labels and routing tables}
Let $v$ be a vertex. The label of $v$ consists of its $x$- and
$y$-coordinate as well as the bounding $x$-coordinates of $I(v)$.
We do not need $v_{\id}$ since $(v_x,v_y)$ identifies the vertex in the network.
Thus, $\Lab(n)=4 \cdot \lceil \log n \rceil$ since we can
assume that $v_x, v_y \in \{0, \dots, n-1\}$.
In the routing table of $v$, we store the 
bounding $x$-coordinates of $I^2(\bd(v))$ as well 
as the bounding $x$-coordinates of
$I^2(\td(v))$. Furthermore, we store $(\bd^2(v)_x,\bd^2(v)_y,bit)$
where $bit$ indicates
whether $\td(v)$ or $\bd(v)$ is on the path $\pi_b(v,2)$. Thus,
$\Tab(n) = 6 \cdot \lceil \log n \rceil +1$.

\subparagraph*{The routing function.}
We are given a current vertex $s$ together with its
routing table and link table, the label of a target vertex $t$, and
a header. If $t \in N(s)$, then $\lab(t)$ is
in the link table of $s$, and we send the data packet directly to $t$.
If the header is non-empty, it will contain the coordinates of exactly one
vertex visible from $s$. We clear the header and go to this respective vertex.
The remaining discussion assumes that the header is empty and 
that $t \not\in N(s)$.
The routing function now distinguishes 
four cases depending on whether $t \in I(s)$,
$t \in I^2(s)$ or $t\in I^3(s)$. We can check the first and the second
condition locally, using the link table of $s$ as well as the 
label of $t$ (note that from the link table of $s$, we can deduce 
$a^*(s)$ and $b^*(s)$, and their interval boundaries).
To check the third condition locally, we use
\cref{lem:I3-equals-exended-of-dominators} which shows that 
$I^3(s) = I^2(\bd(s)) \cup I^2(\td(s))$. Since we
stored the bounding $x$-coordinates of these two intervals in the routing
table of $s$, we can check $t\in I^3(s)$ easily.

\textbf{Case 1 {\boldmath $(t\in I(s)\setminus N(s))$}:}
if $\fd(s,t)$ is a vertex, we can determine it by using the
link table and the label of $t$. The packet is sent to $\fd(s,t)$.
If $\fd(s,t)$ is not a vertex, we determine $\nd(s,t)$
and send the packet there. 
The header remains empty.

\textbf{Case 2 {\boldmath $(t\in I^2(s)\setminus I(s))$}:}
there is an $i\geq 1$ with $t\in[\ell^i,\ell^{i-1}]$ or
$t\in[r^{i-1},r^i]$. We find $i$ using
the link table and $\lab(t)$. The packet is sent
to $a^i$ or $b^i$.
The header remains empty.

\textbf{Case 3 {\boldmath $(t\in I^3(s)\setminus I^2(s))$}:}
if $t\in I^2(\bd(s))$, we send
the packet to $\bd(s)$. Otherwise, $t \in I^2(\td(s))$, and 
we send the packet to 
$\td(s)$. In both cases, the header remains empty.

\textbf{Case 4 {\boldmath $(t\notin I^3(s))$}:}
in the routing table we find the entry $(\bd^2(s)_x,\bd^2(s)_y,bit)$.
We store $(\bd^2(s)_x,\bd^2(s)_y)$ in the header and send
the packet to $\bd(s)$ or $\td(s)$, whichever is indicated by $bit$. 

\subparagraph*{Analysis.}
Obviously, $H(n) = 2 \cdot \lceil \log n \rceil$.
It remains to analyze the stretch factor. 
For this, we show that after
one or two steps, the distance to the target vertex has decreased 
by at least one. This immediately gives a
stretch factor of $2$.

\begin{lemma}
\label[lemma]{lem:routing-scheme-terminates}
Let $s, t \in V(P)$. After at most two steps of the routing scheme 
from $s$ with target label $\lab(t)$, we reach a vertex $v$ with 
$d(v, t) \leq d(s, t) - 1$.
\end{lemma}

\begin{proof}
First, if $t\in N(s)$, then we take one hop and decrease the 
distance to $0$.

Second, suppose that $t\in I(s)\setminus N(s)$. 
If $\fd(s,t)$ is not a vertex, the
next vertex is $\nd(s,t)$, which is on a shortest path from $s$ to 
$t$ due
to \cref{lem:shortest-path-via-dn-or-df}. Otherwise, $\fd(s,t)$ is the next
vertex. If $\fd(s,t)$ is on a shortest path from $s$ to $t$, we are done.
Otherwise, $t$ is not
visible from $\fd(s,t)$, so $\fd^2(s,t)$ has to be the
second vertex on the routed path. By \cref{lem:fd-fd-is-closer}, we have
$d(\fd^2(s,t),t) = d(s,t)-1$.

Third, if $t\in I^2(s)\setminus I(s)$, there is an $i \geq 1$ 
such that the next vertex is either $a^i$ or $b^i$. 
By \cref{lem:t-in-extended-interval},
this vertex is on a shortest path.

Fourth, assume $t\in I^3(s)\setminus I^2(s)$. Let $v_1$ and $v_2$ be 
the next two vertices on the routing path. 
We use \cref{lem:dominators-see-each-other} and
\cref{lem:t-outside-kpp-interval} to conclude $d(v_1,t)\leq d(s,t)$,
as $v_1$ is either $\td(s)$ or $\bd(s)$. Due to the construction
of the routing function, we have $t\in I^2(v_1)\setminus I(v_1)$. Thus, there is
an $i \geq 1$, such that $v_2=a^i(v_1)$ or $v_2=b^i(v_1)$. By
\cref{lem:t-in-extended-interval}, the vertex $v_2$ is on a shortest path from $v_1$
to $t$ and we can conclude $d(v_2,t)=d(v_1,t)-1\leq d(s,t)-1$.

Last, assume $t\notin I^3(s)$. Then, the packet is routed to a vertex
$p \in\{\bd(s),\td(s)\}$, whichever is on a shortest path to
$\bd^2(s)$, and then $\bd^2(s)$.
\cref{lem:dominators-see-each-other} and \cref{lem:t-outside-kpp-interval}
give $d(\bd^2(s), t)\leq d(s,t)-1$.
\end{proof}
A more detailed analysis gives that the label size can be reduced to
$3 \cdot \lceil \log n \rceil - 1$, whereas the routing table size can be reduced
to $5 \cdot \lceil \log n \rceil -2$. However, this will not affect our second main result
which follows from the discussion above.
\begin{restate}[thm:double-routing]
\begin{theorem}[restated]
Let $P$ be a double histogram with $n$ vertices. There is a routing
scheme for $G(P)$ with routing table, label and header size
$O(\log n)$, such that we can route between any two vertices
with a stretch at most $2$.
\end{theorem}
\end{restate}

\section{Conclusion}
We gave the first routing schemes for the hop-distance in simple
polygons. In particular, we have a routing scheme for simple 
histograms with label
size $2 \cdot \lceil \log n \rceil$, routing table size 1, and
stretch 1. We also presented a routing scheme for double
histograms with label, routing table and header size $O(\log n)$
and stretch 2.
This constitutes a first step towards an efficient routing scheme
for the hop-distance in orthogonal polygons. 
The following open problems arise naturally.

First of all, the routing scheme for double histograms 
shows that it is possible to obtain a routing scheme for simple 
histograms with label size $\lceil \log n \rceil$. The stretch factor
increases to $2$. The basic idea is as follows: if $t \in I(s)$, we
determine the far dominator $\fd(s,t)$ and take the hop to
$\fd(s,t)$, without looking at the breakpoint of the near dominator.
Therefore, we save the $\lceil \log n \rceil$ bits that were necessary
to store the $\id$ of the breakpoint. It remains open whether one can
decrease the stretch simultaneously.

As a next step, it would be interesting to see how the routing scheme
extends to monotone polygons as well as
arbitrary orthogonal polygons, assuming $r$-visibility.

After that, it will be interesting to take a closer look at (orthogonal)
polygons assuming $l$-visibility. Here, the structure of visibility -- even in
simple histograms -- is much more complicated. Moreover, we can no longer 
assume integer coordinates.

Last but not least, it would be interesting to know, whether it is possible
to decrease the stretch in double histograms to, say $1 + \eps$, for $\eps > 0$. 

\bibliographystyle{plainurl}
\bibliography{sources.bib}

\end{document}